\documentclass[12pt]{article}

\usepackage[inner=32mm,outer=31mm,tmargin=30mm,bmargin=40mm]{geometry}

\usepackage{latexsym,amsfonts,amsmath,amssymb,epsfig,tabularx,amsthm,dsfont,mathrsfs}

\usepackage{graphicx}
\usepackage{enumerate}

\usepackage{titling}
\newcommand{\subtitle}[1]{%
  \posttitle{%
    \par\end{center}
    \begin{center}\large#1\end{center}
    \vskip0.5em}%
}

\usepackage{hyperref} 
\hypersetup{colorlinks=true,
        linkcolor=black,
        citecolor=black,
        urlcolor=blue}

\usepackage{pgfplots}

\usetikzlibrary{external}
\newcommand{\bse}{{\boldsymbol{e}}}
\newcommand{\bsx}{{\boldsymbol{x}}}
\newcommand{\bsy}{{\boldsymbol{y}}}
\newcommand{\bsw}{{\boldsymbol{w}}}
\newcommand{\bsz}{{\boldsymbol{z}}}
\newcommand{\bsalpha}{{\boldsymbol{\alpha}}}
\newcommand{\bsbeta}{{\boldsymbol{\beta}}}
\newcommand{\bsgamma}{{\boldsymbol{\gamma}}}

\newcommand{\rd}{{\mathrm{d}}}
\newcommand{\R}{{\mathbb{R}}} 

\newtheorem{algo}{Algorithm}{\upshape \bfseries}{\itshape}
\newtheorem{definition}{Definition}
\newtheorem{remark}{Remark}
\newtheorem{theorem}{Theorem}
\newtheorem{lemma}{Lemma}
\newtheorem{corollary}{Corollary}

\DeclareMathAlphabet{\mathup}{OT1}{\familydefault}{m}{n}

\newcommand{\widebar}[1]{\mbox{\kern1.5pt\hbox{\vbox{\hrule height 0.6pt \kern0.35ex
        \hbox{\kern-0.15em \ensuremath{#1 }\kern0.0em}}}}\kern-0.1pt}

\newlength{\fixboxwidth}
\usepackage{color}
\usepackage{soul}

\usepackage{latexsym,amsfonts,amsmath,amssymb,mathrsfs}
\usepackage{url}
\usepackage{graphicx}
\usepackage{color}
\usepackage{euscript}
\usepackage{verbatim}
\usepackage{hyperref}

\begin{document}

\title{A weighted Discrepancy Bound of quasi-Monte Carlo Importance Sampling}

\author{Viacheslav Natarovskii\thanks{Institute for Mathematical Stochastics, 
Georg-August-Universit\"at G\"ottingen, Goldschmidtstra\ss e 7, 37077 G\"ottingen, 
Email: vnataro@uni-goettingen.de, daniel.rudolf@uni-goettingen.de, bjoern.sprungk@uni-goettingen.de}, 
Daniel Rudolf$^{\,*,}$\thanks{Felix-Bernstein-Institute for Mathematical Statistics 
 in the Biosciences, Goldschmidtstra\ss e 7, 37077 G\"ottingen}, 
 Bj\"orn Sprungk$^*$}

\author{Josef Dick\thanks{The University of New South Wales, Sydney, NSW 2052, Australia, Email: josef.dick@unsw.edu.au},
Daniel Rudolf\thanks{Institute for Mathematical Stochastics,
	Universit\"at G\"ottingen, Goldschmidtstra\ss e 7,
	37077 G\"ottingen,  Germany, Email: daniel.rudolf@uni-goettingen.de},
	Houying Zhu\thanks{Melbourne Integrative Genomics \& School of Mathematics and Statistics,
		The University of Melbourne, Parkville, VIC 3010, Australia, Email: houying.zhu@unimelb.edu.au}
}


\date{\today}

\maketitle
\begin{abstract}
Importance sampling Monte-Carlo methods are widely used for the approximation of expectations with respect to partially known probability measures. In this paper we study a deterministic version of such an estimator based on quasi-Monte Carlo. We obtain an explicit error bound in terms of the star-discrepancy for this method.
\end{abstract}


{\bf Classification. Primary: 62F15; Secondary: 11K45.} 

\section{Introduction}


In statistical physics and Bayesian statistics it is desirable to compute
expected values 
\begin{equation}  \label{eq: mean}
\mathbb{E}_\pi (f) = \int_{\mathbb{R}^d} f(\bsx)\, \rd \pi(\bsx)
\end{equation}
with $f\colon \mathbb{R}^d \to \R$ and a partially known probability measure $\pi$
on $(\mathbb{R}^d,\mathcal{B}(\mathbb{R}^d))$. 
Here $\mathcal{B}(\mathbb{R}^d)$ denotes the Borel $\sigma$-algebra and 
partially known means
that there is an unnormalized density 
$u \colon \mathbb{R}^d \to [0,\infty)$ (with respect to the Lebesgue measure)
and $\int_{\mathbb{R}^d} u(\bsx) \,\rd \bsx  \in (0,\infty)$, 
such that 
\begin{equation}  \label{eq: pi}
\pi(A) = \frac{\int_A u(\bsx)\, \rd \bsx }{\int_{\mathbb{R}^d} u(\bsy)\, \rd \bsy }, \qquad
A\in \mathcal{B}(\mathbb{R}^d).
\end{equation}


Probability measures of this type are met in numerous applications.
For example, for the density of a Boltzmann distribution one has
\[
u(\bsx) = \exp(-\beta H(\bsx)), \quad \bsx\in \mathbb{R}^d,
\]
with
inverse temperature $\beta>0$ and
Hamiltonian $H\colon \mathbb{R}^d \to \mathbb{R}$. 
The density of a posterior distribution is also of this form. 
Given observations $\bsy \in \mathcal{Y}$, likelihood
function $\ell(\bsy \mid \bsx)$ and prior probability density $p$, 
with respect to the Lebesgue measure on $\mathbb{R}^d$,
\[
u(\bsx) = \ell(\bsy \mid \bsx)\, p(\bsx), \quad \bsx\in \mathbb{R}^d.
\]
In this setting $\mathbb{R}^d$ is considered as parameter- and $\mathcal{Y}$ as observable-space.
In both examples, the
normalizing constant is in general unknown.


In the present work we only consider unnormalized densities $u$ which are zero 
outside of the unit cube $[0,1]^d$. Hence we restrict ourself to
$u\colon [0,1]^d \to [0,\infty)$, i.e., $\pi$ is a probability measure
on $[0,1]^d$, and $f\colon [0,1]^d \to \mathbb{R}$. 
To stress the dependence on the unnormalized density in \eqref{eq: mean},
define
\[
S(f,u) := 
\frac{\int_{[0,1]^d} f(\bsx) u(\bsx)\, \rd \bsx}{\int_{[0,1]^d} u(\bsy)\,\rd \bsy} =  \mathbb{E}_{\pi}(f)
\]
for $f$ and $u$ belonging to some class of functions.
It is desirable to have algorithms which approximately compute $S(f,u)$
by only having access to function values of $f$ and $u$
without knowing the normalizing constant a priori.
A straightforward strategy to do so provides an importance sampling Monte Carlo approach.
It works as follows.
\begin{algo} \label{alg: MC}
	Monte Carlo importance sampling:
	\begin{enumerate}
		\item Generate a sample of an i.i.d. sequence of random variables
		$X_1,\dots,X_n$ with $X_i \sim \mathcal{U}([0,1]^d)$\footnote{By $\mathcal{U}([0,1]^d)$ we denote the uniform distribution 
			on $[0,1]^d$.} and call the result
		$\bsx_1,\dots,\bsx_n$.
		\item Compute 
		\[
		M_n(f,u) := \frac{\sum_{j=1}^n f(\bsx_j) u(\bsx_j)}{\sum_{j=1}^n u(\bsx_j)}.
		\]
	\end{enumerate}
\end{algo}
Under the minimal assumption that $S(f,u)$ is finite, 
a strong law of large numbers argument guarantees that the importance sampling
estimator $M_n(f,u)$ is well-defined, cf. \cite[Chapter 9, Theorem~9.2]{Ow13}.
For uniformly bounded $f$ and finite $\sup u/ \inf u$ an 
explicit error bound of the mean square error is proven in \cite[Theorem~2]{MaNo07}.

Surprisingly, there is not much known about a deterministic version of this method.
The idea is
to substitute 
the uniformly in $[0,1]^d$ distributed i.i.d. sequence 
by a carefully chosen deterministic point set. 
Carefully chosen in the sense that the point set
$P_n=\{\bsx_1,\dots,\bsx_n\} \subset [0,1]^d$ has
``small'' star-discrepancy, that is,
\[
D_{\lambda_d}(P_n) := \sup_{\bsx\in [0,1]^d } 
\left \vert \frac{1}{n} \sum_{j=1}^n \mathbf{1}_{[0,\bsx)}(\bsx_j) - \lambda_d([0,\bsx)) \right \vert 
\]
is ``small''. 
Here, 
the set 
$
[0,\bsx) = \prod_{i=1}^d [0,x_i)
$ 
denotes an anchored box in $[0,1]^d$ 
with $\bsx=(x_1,\dots,x_d)$ and $\lambda_d([0,\bsx))=\prod_{i=1}^d x_i $ 
is the $d$-dimensional Lebesgue measure of $[0,\bsx)$. 
This leads to a quasi-Monte Carlo importance sampling method.
\begin{algo} \label{qmc}
	Quasi-Monte Carlo importance sampling:
	\begin{enumerate} 
		\item Generate a point set $P_n=\{\bsx_1,\dots,\bsx_n\}$ with ``small'' star discrepancy $D_{\lambda_d}(P_n)$. 
		\item Compute 
		\begin{equation}  \label{eq: I_n}
		Q_n(f,u) = \frac{\sum_{j=1}^n f(\bsx_j) u(\bsx_j)}{\sum_{j=1}^n u(\bsx_j)}.
		\end{equation}
	\end{enumerate}
\end{algo}
Our main result, stated in Theorem~\ref{thm: main}, is an explicit error bound 
for the estimator $Q_n$ of the form
\begin{equation}  \label{main_est}
\vert S(f,u) - Q_n(f,u) \vert \leq 4
\frac{\Vert f \Vert_{H_1} \Vert u \Vert_D}{\int_{[0,1]^d} u(\bsx)\rd \bsx} \; D_{\lambda_d}(P_n).
\end{equation}
Here $f$ must be differentiable, such that $\Vert f \Vert_{H_1}$, 
defined in \eqref{eq: H_1_norm} below, 
is finite. As a regularity assumption on $u$ it is assumed that $\Vert u \Vert_D$, 
defined in \eqref{eq: u_D} below,
is also finite. 

The estimate of \eqref{main_est} is proven by two results which might be interesting on its own. 
The first is a Koksma-Hlawka inequality
in terms of a weighted star-discrepancy, see Theorem~\ref{thm_int_error}. 
The second is a relation
between this 
quantity and the classical star-discrepancy, see Theorem~\ref{weightedDiscr}.
To illustrate the quasi-Monte Carlo importance sampling procedure and the error bound
we provide an example in Section~\ref{sec: Ill_ex} where \eqref{main_est} is applicable.

\emph{Related Literature.} The Monte Carlo importance sampling procedure 
from Algorithm~\ref{alg: MC} is well studied. In \cite{MaNo07}, Novak and Math{\'e} prove
that it is optimal on a certain class of tuples $(f,u)$. 
However, recently this Monte Carlo approach attracted considerable attention, 
let us mention here \cite{AgPaSaSt15,ChDi15}. In particular, in \cite{AgPaSaSt15} upper error bounds 
not only for bounded functions $f$ are provided and the relevance of the method 
for inverse problems is presented. 

Another standard approach the approximation of $\mathbb{E}_{\pi}(f)$ are
Markov chain Monte Carlo methods.
For details concerning error bounds we refer to \cite{Jo04,JoOl10,LaMiNi09,Pa16,Ru09,Ru10,Ru12} 
and the references therein. Combinations of importance sampling and Markov chain Monte Carlo
are for example analyzed in \cite{RoTaFl15,ViHeFr16,RuSp2018}.

The quasi-Monte Carlo importance sampling procedure of Algorithm~\ref{qmc} is, to our knowledge, 
less well studied.
An asymptotic convergence result is stated in \cite[Theorem~1]{ChGe14} 
and promising numerical experiments
are conducted in \cite{HoLe05}. 
A related method, a randomized deterministic sampling procedure according to the unnormalized
distribution $\pi$, is studied in \cite{CoVa10}. Recently, \cite{ACHLT18} explore the efficiency of using 
QMC inputs in importance sampling for Archimedean copulas 
where significant variance reduction is obtained for a case study.

A quasi-Monte Carlo approach to Bayesian inversion was used in 
\cite{DiGaLGiSc16}
and in \cite{DiGaLGiSc17}
The latter paper uses a combination of quasi-Monte Carlo and the multi-level method. 
The computation of the likelihood function involves solving a partial differential equation, 
but otherwise the problem is of the same form as described in the introduction. 


\section{Weighted Star-discrepancy and error bound}
Recall that $[0,\bsx)$ for $\bsx\in [0,1]^d$ are boxes anchored at $0$.
As a measure of ``closeness'' between the empirical distribution
$\frac{1}{n} \sum_{j=1}^n \mathbf{1}_{[0,\bsx)}(\bsx_i)$
of a point set $P_n=\{\bsx_1,\dots,\bsx_n\}$ to $\lambda_d([0,\bsx))$ 
we consider the star-discrepancy $D_{\lambda_d}(P_n)$. 
A straightforward extension 
of this quantity taking the probability measure $\pi$ on $[0,1]^d$ into account
is the following weighted discrepancy.

\begin{definition}[Weighted Star-discrepancy]
	For a given point set 
	$P_n=\{\bsx_1,\dots,\bsx_n\}\subset [0,1]^d$ and weight vector 
	$\bsw=(w_1,\dots,w_n)\in \mathbb{R}^n$, 
	which might depend on $P_n$ and satisfies $\sum_{i=1}^n w_i=1$,
	define  
	the \emph{weighted star-discrepancy} by
	\[
	D_\pi(\bsw,P_n) = \sup_{\bsx\in [0,1]^d} 
	\left\vert \sum_{i=1}^n w_i \mathbf{1}_{[0,\bsx)}(\bsx_i) - \pi([0,\bsx))\right \vert.
	\]
\end{definition}

\begin{remark}
	If $\pi$ is the Lebesgue measure on $[0,1]^d$ and the weight vector is 
	$\bsw=(1/n,\dots,1/n)$, then $D_{\lambda_d}(P_n) = D_\pi(\bsw,P_n)$ for any point set $P_n$. 
	For general 
	$\pi$ with unnormalized density $u\colon [0,1]^d\to [0,\infty)$, 
	allowing the representation
	\eqref{eq: pi}, we focus on the weight vector
	\begin{equation} \label{eq: weight_u}
	w^u_i := w_i(u,P_n) :=  \frac{u(\bsx_i)}{\sum_{j=1}^n u(\bsx_j)}, \qquad i =1,\dots,n. 
	\end{equation}
	Here let us emphasize that $\bsw^u:=(w^u_1,\dots,w^u_n)$ depends on $u$ and 
	$P_n$.
\end{remark}

\subsection{Integration Error and weighted Star-discrepancy}
With standard techniques one can prove a Koksma-Hlawka inequality 
according to $D_{\pi}(w,P_n)$. 
For details
we refer to \cite{DiHiPi14}, \cite[Section~2.3]{DiPi10} and \cite[Chapter~9]{NoWo10}.
A similar inequality of a quasi-Monte Carlo importance 
sampler 
can be found in 
\cite[Corollary~1]{AiDi15}. 

Let $[d]:=\{1,\dots,d\}$
and 
$L_2([0,1]^d)$ be the space of square integrable functions with respect to
the Lebesgue measure.
Define the reproducing kernel $K \colon [0,1]^d \times [0,1]^d \to [0,1]$ by
$
K(\bsx,\bsy):= \prod_{i=1}^d (1+ \min\{1-x_i,1-y_i\}).
$
By $H_2=H_2(K)$ we denote the corresponding reproducing kernel Hilbert space, 
which consists of differentiable functions with respect to all
variables with first partial derivatives being in $L_2([0,1]^d)$. 
For $f,g\in H_2$ the inner product is given by
\[
\langle f,g \rangle 
= \sum_{v\subseteq [d]}  \int_{[0,1]^{\vert v \vert}} 
\frac{\partial^{\vert v \vert}}{\partial \bsx_v} f(\bsx_v;1)  
\frac{\partial^{\vert v \vert}}{\partial \bsx_v} g(\bsx_v;1) 
\;\rd \bsx_v,
\]
where for $v\subseteq [d]$ and $\bsx=(x_1,\dots,x_d)$ we write $\bsx_v = (x_j)_{j\in v}$
and $(\bsx_v;1)=(z_1,\dots,z_d)$ with
$z_j = x_j$ if $j\in V$ and $z_j=1$ if $j\not\in v$.
Thus, $H_2$
consists of functions 
which are differentiable according to all
variables with first partial derivatives being in $L_2([0,1]^d)$. 
Note that, for $v\subseteq [d]$ holds 
\[
\frac{\partial^{\vert v \vert}}{\partial \bsx_v} K((\bsx_v;1),\bsy)
= (-1)^{\vert v \vert} \mathbf{1}_{[\bsy_v,1]}(\bsx_v),
\]
where $[\bsy_v,1] = \prod_{i\in v} [y_i,1]$ with $\bsy = (y_1,\dots,y_d)\in [0,1]^d$.
Thus, the reproducing property of the reproducing kernel Hilbert space can be
rewritten as
\begin{equation}  \label{eq: repr_prop}
f(\bsy) = \sum_{v\subseteq [d]} \int_{[\bsy_v,1]} 
(-1)^{\vert v \vert} \frac{\partial^{\vert v\vert} }{\partial \bsx_{v}}f(\bsx_v;1) \rd \bsx_v.
\end{equation}
Further, we define the space $H_1$ of differentiable functions 
$f\colon [0,1]^d \to \mathbb{R}$ 
with finite norm
\begin{equation}  \label{eq: H_1_norm}
\Vert f \Vert_{H_1} := \sum_{v\subseteq [d]} \int_{[0,1]^{\vert v \vert}} 
\left|\frac{\partial ^{\vert v \vert} }{\partial \bsx_v}  f(\bsx_v;1)\right| \rd \bsx_v,
\end{equation}
where for $v = \emptyset$ we have $ \int_{[0,1]^{\vert v \vert}} 
\left|\frac{\partial ^{\vert v \vert} }{\partial \bsx_v}  f(\bsx_v;1)\right| \rd \bsx_v = | f(1) |$.
We also define the semi-norm
\begin{equation}  \label{eq: H_1_seminorm}
\Vert f \Vert_{\widetilde{H}_1} := \sum_{\emptyset \neq v\subseteq [d]} \int_{[0,1]^{\vert v \vert}} 
\left|\frac{\partial ^{\vert v \vert} }{\partial \bsx_v}  f(\bsx_v;1)\right| \rd \bsx_v.
\end{equation}
It is obvious that $\Vert f\Vert_{\widetilde{H}_1} \le \Vert f \Vert_{H_1}$.

We have the following relation between the integration error in $H_1$ and the weighted discrepancy.

\begin{theorem}[Koskma-Hlawka inequality]\label{thm_int_error}
	Let $\pi$ be a 
	probability measure of the form \eqref{eq: pi} with unnormalized density $u\colon [0,1]^d \to [0,\infty)$.
	Then, for $P_n=\{\bsx_1, \ldots, \bsx_n \} \subset [0,1]^d$, 
	arbitrary weight vector $\bsw=(w_1,\dots,w_n)\in \mathbb{R}^n$ with $\sum_{i=1}^d w_i=1$,
	and for all $f \in H_1$ we have
	\begin{equation*}
	\left|
	S(f,u)
	- \sum_{i=1}^n w_i f(\bsx_i)
	\right| 
	\le \|f\|_{\widetilde{H}_1} \,D_{\pi}(\bsw,P_n).
	\end{equation*}
\end{theorem}

\begin{proof}
	Define the quadrature error 
	$
	e(f,P_n) := \int_{[0,1]^d} f(\bsx)\, \rd  \pi(\bsx) -  \sum_{i=1}^n w_i f(\bsx_i)
	$
	of the approximation of $\mathbb{E}_\pi(f)=S(f,u)$ 
	by $ \sum_{i=1}^n w_i f(\bsx_i)$. Define the function $\widetilde{f} = f - f(1)$. Then $\widetilde{f}(1) = 0$, $e(f,P_n) = e(\widetilde{f}, P_n)$ and $\Vert f\Vert_{\widetilde{H}_1} = \Vert \widetilde{f} \Vert_{H_1}$.
	
	For
	\[
	h(\bsx) := \int_{[0,1]^d} K(\bsx, \bsy)\, \rd \pi(\bsy) - \sum_{i=1}^n w_i\, K(\bsx, \bsx_i),
	\]
	and $v\subseteq[d]$ 
	we have
	$
	\frac{\partial^{\vert v \vert}}{\partial \bsx_v} h(\bsz_v;1) 
	= (-1)^{\vert v \vert} \left ( \pi([0,(\bsz_v;1)))-\sum_{i=1}^n w_i \mathbf{1}_{[0,\bsz_v]}(\bsx_{i,v})\right ).
	$
	A straightforward calculation, see also for instance \cite[formula (3)]{DiHiPi14},
	shows by using \eqref{eq: repr_prop}  
	that
	\begin{align*}
	e(\widetilde{f},P_n) 
	& = \sum_{v\subseteq [d]}
	\int_{[0,1]^{\vert v \vert}} 
	\frac{\partial ^{\vert v \vert} }{\partial \bsx_v}  \widetilde{f}(\bsz_v;1)
	(-1)^{\vert v \vert} 
	\left ( \pi([0,(\bsz_v;1)))-\sum_{i=1}^n w_i \mathbf{1}_{[0,\bsz_v]}(\bsx_{i,v})\right ) \rd \bsz\\
	& = \langle \widetilde{f},h \rangle.
	\end{align*}
	Finally, by 
	$ \left\vert\frac{\partial^{\vert v \vert}}{\partial \bsz _v} h(\bsz_v;1)\right\vert \leq D_\pi(\bsw,P_n)$
	we have
	$$
	|e(f,P_n)| = |e(\widetilde{f},P_n)|  \leq \Vert \widetilde{f} \Vert_{H_1}  D_{\pi}(\bsw,P_n) = \Vert f\Vert_{\widetilde{H}_1} D_{\pi}(\bsw, P_n),
	$$
	which finishes the proof.
\end{proof}
An immediate consequence of the theorem with $\bsw^u$ from \eqref{eq: weight_u} and $Q_n$
from \eqref{eq: I_n} is the error bound
%
%
\begin{equation*}
\left|
S(f,u)
- 
Q_n(f,u)
\right| 
\le \|f\|_{H_1} \; D_{\pi}(\bsw^u,P_n).
\end{equation*}
Here the dependence on $u$ on the right-hand side is hidden in $D_{\pi}(\bsw^u,P_n)$ through $\bsw^u$ and $\pi$.
The intuition is, that under suitable assumptions on $u$ the weighted star-discrepancy
can be bounded by the classical star-discrepancy of $P_n$.

\subsection{Weighted and classical Star-discrepancy}
In this section we provide a relation between the classical star-discrepancy $D_{\lambda_d}(P_n)$
and the weighted star-discrepancy $D_{\pi}(\bsw^u,P_n)$.

\begin{theorem}\label{weightedDiscr} 
	Let $\pi$ be a probability measure of the form \eqref{eq: pi} with unnormalized 
	density function $u \colon [0,1]^d \to [0,\infty)$.
	Then, for any point set $P_n=\{\bsx_1,\ldots,\bsx_{n}\}$ in $[0,1]^d$,
	we have
	\begin{equation*}
	D_{\pi}(\bsw^u,P_n) \le 4 D_{\lambda_d}(P_n) 
	\frac{\left \Vert u \right \Vert_{D}}{\int_{[0,1]^d} u(\bsx)\rd \bsx},
	\end{equation*} 
	where 
	\begin{equation}  \label{eq: u_D}
	\left \Vert u \right \Vert_{D} 
	= \sup_{\bsz \in [0,1]^d} u(\bsz) +  \sup_{\bsz \in [0,1]^d}  \left \Vert  u(T_{\bsz}\, \cdot) \right \Vert_{\widetilde{H}_1} 
	\end{equation}
	with $T_{\bsz}\colon [0,1]^d \to [0,\bsz]$ 
	and $T_{\bsz}(x_1,\dots,x_d) = (z_1 x_1,\dots,z_d x_d)$ for $\bsz\in [0,1]^d$.
\end{theorem}

\begin{proof} 
	For the given point set $P_n\subset [0,1]^d$ and unnormalized density $u$
	recall that $\bsw^u$ is defined in \eqref{eq: weight_u}. 
	To shorten the notation define $\Vert u \Vert_1 := \int_{[0,1]^d} u(\bsy) \rd \bsy$.
	Then, for $\bsz\in [0,1]^d$ we have
	\begin{align*}
	& \left \vert \sum_{j=1}^n w^u_j\mathbf{1}_{[0,\bsz)}(\bsx_j) -\pi([0,\bsz)) \right \vert 
	= \left \vert  \frac{\sum_{j=1}^n u(\bsx_j)\mathbf{1}_{[0,\bsz)}(\bsx_j)}{\sum_{i=1}^n u(\bsx_i)} 
	-\frac{\int_{[0,\bsz)} u(\bsx) \rd \bsx}{ \Vert u \Vert_1} \right \vert\\
	& \quad\leq   \frac{\sum_{j=1}^n u(\bsx_j)\mathbf{1}_{[0,\bsz)}(\bsx_j)}
	{\Vert u \Vert_1 \sum_{i=1}^n u(\bsx_i)} 
	\left \vert \Vert u \Vert_1 - \frac{1}{n} \sum_{i=1}^n u(\bsx_i) \right \vert\\
	& \qquad\qquad\qquad + 
	\frac{1}{\Vert u \Vert_1} 
	\left \vert \frac{1}{n} \sum_{i=1}^n u(\bsx_i) \mathbf{1}_{[0,\bsz)}(\bsx_i)
	- \int_{[0,\bsz)} u(\bsx) \rd \bsx
	\right \vert \\
	&\quad \leq  \frac{2}{\Vert u \Vert_1} \sup_{z\in [0,1]^d}
	\left \vert \frac{1}{n} \sum_{i=1}^n u(\bsx_i) \mathbf{1}_{[0,\bsz)}(\bsx_i)
	- \int_{[0,\bsz)} u(\bsx) \rd \bsx \right \vert.
	\end{align*}
	For $\bsz\in [0,1]^d$ denote $P^{\bsz}=P_n\cap [0,\bsz)$ 
	and let $\left \vert P^{\bsz} \right \vert$ be 
	the cardinality of $P^{\bsz}$. Define
	\begin{align*}
	I_1(\bsz) & := \frac{\int_{[0,\bsz)}u(\bsx) \rd \bsx}{\lambda_d([0,\bsz))}
	\left \vert \frac{\vert P^{\bsz} \vert}{n} - \lambda_d([0,\bsz))  \right \vert, \\
	I_2(\bsz) & := \frac{\vert P^{\bsz} \vert}{n} 
	\left \vert \frac{1}{\vert P^{\bsz} \vert} \sum_{\bsx \in P^{\bsz}} u(x) - 
	\frac{\int_{[0,\bsz)} u(\bsx) \rd \bsx  }{\lambda_d([0,\bsz))}\right \vert,
	\end{align*}
	and note that
	\begin{align*}
	&  \left \vert \frac{1}{n} \sum_{i=1}^n u(\bsx_i) \mathbf{1}_{[0,\bsz)}(\bsx_i)
	- \int_{[0,\bsz)} u(\bsx) \rd \bsx \right \vert\\
	& = \frac{\vert P^{\bsz} \vert}{n} 
	\left\vert 
	\frac{1}{\vert P^{\bsz}  \vert}\sum_{x\in P^{\bsz}} u(x) 
	- \frac{n}{\vert P^{\bsz} \vert}\int_{[0,\bsz)}u(\bsx) \rd \bsx
	\right\vert 
	\leq I_1(\bsz) + I_2(\bsz).
	\end{align*}
	{\bf Estimation of $I_1(\bsz)$:}
	An immediate consequence of the 
	definition of $I_1(\bsz)$ is
	\begin{equation}\label{I_1}
	I_1(\bsz) \leq \frac{\int_{[0,\bsz]} u(\bsx) \rd \bsx}{\lambda_d([0,\bsz])}\;
	D_{\lambda_d}(P_n) \leq D_{\lambda_d}(P_n) \sup_{\bsx \in [0, \bsz]} u(\bsx).
	\end{equation}
	{\bf Estimation of $I_2(\bsz)$:}
	With the transformation $T_{\bsz}\colon [0,1]^d \to [0,\bsz]$ defined in 
	the theorem one has
	$
	\frac{\int_{[0,\bsz]} u(\bsx) \rd \bsx}{\lambda_d([0,\bsz])}
	= \int_{[0,1]^d} u(T_{\bsz}\bsx)\, \rd \bsx.
	$
	Let
	\[
	Q:=T^{-1}_{\bsz}\, P^{\bsz} = \{ (z_1^{-1} x_1 ,\dots,z_d^{-1} x_d )\mid \bsx\in P^{\bsz} \} 
	\subset [0,1]^d
	\]
	and observe that $\vert P^{\bsz} \vert = \vert Q \vert$.
	Then
	\begin{align*}
	I_2(\bsz) = \frac{|P^{\bsz}|}{n}\left\vert \frac{1}{|Q|}\sum_{ x \in Q} u(T_{\bsz} \, \bsx)
	-\int_{[0,1]^d} u(T_{\bsz}\,\bsx)\,\rd \bsx\right| 
	\leq \frac{\left \vert P^{\bsz} \right \vert}{n} D_{\lambda_d}(Q) 
	\left \Vert u(T_{\bsz}\,\cdot) \right \Vert_{H_1},
	\end{align*} 
	where the last inequality follows from Theorem~\ref{thm_int_error}
	with $\bsw = (1/n,\dots,1/n)$ and constant unnormalized density.
	Further,
	\begin{align*}
	\frac{\vert P^{\bsz}\vert}{n} D_{\lambda_d}(Q)
	& = \frac{\vert P^{\bsz}\vert}{n} \sup_{\bsy \in [0,1]^d}
	\left \vert \frac{1}{\vert Q \vert} \sum_{x\in Q} \mathbf{1}_{[0,\bsy)}(x) - \lambda_d([0,\bsy)) \right \vert\\
	& = \sup_{\bsy \in [0,1]^d}
	\left \vert \frac{1}{n} \sum_{x\in Q} \mathbf{1}_{[0,\bsy)}(x) - \frac{\vert Q \vert}{n}\lambda_d([0,\bsy)) \right \vert\\
	& = \sup_{\bsy \in [0,1]^d}
	\left \vert \frac{1}{n} \sum_{x\in P_n} \mathbf{1}_{T_{\bsz}([0,\bsy))}(x) - \frac{\vert Q \vert}{n}\lambda_d([0,\bsy)) \right \vert\\
	& \leq \sup_{\bsy \in [0,1]^d}
	\left \vert \frac{1}{n} \sum_{x\in P_n} 
	\mathbf{1}_{T_{\bsz}([0,\bsy))}(x) - \lambda_d(T_{\bsz}([0,\bsy))) \right \vert\\
	& \qquad \qquad + \sup_{\bsy \in [0,1]^d}
	\left \vert \lambda_d(T_{\bsz}([0,\bsy)))
	- \frac{\vert Q \vert}{n}\lambda_d([0,\bsy)) \right \vert.
	\end{align*}
	By the fact that $T_{\bsz}([0,\bsy))$ is again a box anchored at $0$ and
	\begin{align*}
	\sup_{\bsy \in [0,1]^d}
	\left \vert \lambda_d(T_{\bsz}([0,\bsy)))
	- \frac{\vert Q \vert}{n}\lambda_d([0,\bsy)) \right \vert
	= & \sup_{\bsy \in [0,1]^d} \lambda_d([0,\bsy))
	\left \vert \lambda_d([0,\bsz))
	- \frac{\vert P^{\bsz} \vert}{n} \right \vert \\ \le & \left \vert \lambda_d([0,\bsz))
	- \frac{\vert P^{\bsz} \vert}{n} \right \vert,
	\end{align*}
	we have 
	\[
	I_2(\bsz) \le 2 \left \Vert u(T_{\bsz}\,\cdot) \right \Vert_{\widetilde{H}_1}  \left \vert \lambda_d([0,\bsz))
	- \frac{\vert P^{\bsz} \vert}{n} \right \vert \le 2  \left \Vert u(T_{\bsz}\,\cdot) \right \Vert_{\widetilde{H}_1} D_{\lambda_d}(P_n).
	\]
	
	Hence we have
	\begin{align*}
	& \sup_{\bsz \in [0,1]^d} \left \vert \sum_{j=1}^n w^u_j\mathbf{1}_{[0,\bsz)}(\bsx_j) -\pi([0,\bsz)) \right \vert
	\le   2 \sup_{\bsz \in [0,1]^d} I_1(\bsz) + I_2(\bsz) \\ 
	\le & 2 D_{\lambda_d}(P_n) \sup_{\bsz \in [0,1]^d}  \left(\sup_{\bsx \in [0,\bsz]} u(x) +   2 \left \Vert u(T_{\bsz}\,\cdot) \right \Vert_{\widetilde{H}_1} \right),
	\end{align*}
	which implies the result.
\end{proof}

In particular, the theorem implies that whenever $\Vert u \Vert_D $ is finite and
$D_{\lambda_d}(P_n)$ goes to zero as $n$ 
goes to infinity, also $D_{\pi}(\bsw^u,P_n)$ goes to zero for increasing $n$ with the same rate of
convergence.

\subsection{Explicit error bound}

An immediate consequence of the results of the previous two sections is the following explicit error bound
of the quasi-Monte Carlo importance sampling method of Algorithm~\ref{qmc}.

\begin{theorem}  \label{thm: main}
	Let $\pi$ be a 
	probability measure of the form \eqref{eq: pi} with unnormalized density $u\colon [0,1]^d \to [0,\infty)$.
	Then, for any point set $P_n=\{\bsx_1,\ldots,\bsx_{n}\}$ in $[0,1]^d$, $f\in H_1$ and $Q_n$ from \eqref{eq: I_n} 
	we obtain
	\[
	\vert S(f,u) - Q_n(f,u) \vert \leq 
	4 \frac{\Vert f \Vert_{H_1} \Vert u \Vert_D}{\int_{[0,1]^d} u(\bsx)\rd \bsx} \; D_{\lambda_d}(P_n),
	\]
	with $\Vert u \Vert_D$ from Theorem~\ref{weightedDiscr}.
\end{theorem}

Under the regularity assumption that $\Vert u \Vert_D$ is finite, the error bound tells us that the
classical star-discrepancy determines the rate of convergence on how fast $Q_n(f,u)$ 
goes to $S(f,u)$.

\section{Illustrating Example} \label{sec: Ill_ex}

Define the $d$-simplex by
$
\Delta_d:= \left\{\bsx \in [0,1]^d \colon \sum_{i=1}^d x_i \leq 1 \right\}
$
and consider the (slightly differently formulated) unnormalized density $u\colon [0,1]^d \to [0,1)$ of 
the Dirichlet distribution with parameter vector $\bsalpha \in (1,\infty)^{d+1}$
given by
\begin{equation}
\label{eq: dirichlet_density}
u(\bsx;\bsalpha) = \begin{cases}
(1-\sum_{i=1}^{d} x_i)^{\alpha_{d+1}-1}\prod_{i=1}^d x_i^{\alpha_i-1},   
& x\in \Delta_d,\\
0, & x\not \in \Delta_d.
\end{cases}
\end{equation}
The Dirichlet distribution is the conjugate prior of the multinomial distribution:
Assume that we observed some data $\bsy = (y_1,\dots,y_{d+1}) \in [0,\infty)^{d+1}$, which
we model as a realization of a multinomial distributed 
random variable with unknown parameter vector  $\bsx = (x_1,\dots,x_d)\in [0,1]^d$.
With $n\in\mathbb{N}$ this leads to a likelihood function
$
\ell(\bsy \mid \bsx) = \frac{n!}{ y_1! \cdots y_{d+1}!} (1-\sum_{i=1}^d x_i)^{y_{d+1}} \prod_{i=1}^d x_i^{y_i}.
$
For a prior distribution 
with unnormalized density $u(\bsx,\bsbeta)$ and  
$\bsbeta\in (1,\infty)^{d+1}$ 
we obtain a posterior measure with unnormalized density $u(\bsx,\bsbeta+\bsy)$.

The normalizing constant of $u$ can be computed explicitly, it is known that
\begin{equation} \label{eq: norm_const_dirich}
\int_{[0,1]^d} u(\bsx,\bsalpha)\rd \bsx
= \frac{\prod_{i=1}^{d+1} \Gamma(\alpha_i)}{\Gamma(\sum_{i=1}^{d+1}\alpha_i)}.
\end{equation}
To have a feasible setting for the application of 
Theorem~\ref{thm_int_error} and Theorem~\ref{weightedDiscr} we need to 
show that $\Vert u \Vert_D$ is finite. 
This is not immediately clear, since in $\Vert u \Vert_D$
we take the supremum over $\bsz\in [0,1]^d$.
The following lemma is useful.
\begin{lemma}
	Let $v\subseteq [d]$ and recall that we write $k_v = (k_i)_{i\in v}$.
	Define $(k_v;0;k_{d+1}) = (r_1,\dots,r_{d+1})$ with
	$r_j =  k_j$ if  $j\in v$, $r_j=0$ if $j\not\in v$ and $r_j=k_{d+1}$ if $j=d+1$.
	Assume that $\alpha_i\geq 2$ for $1\leq i\leq d$ and $\alpha_{d+1}\geq d$ .
	Then
	\begin{align}
	\label{al: partial_der}
	\frac{\partial^{\vert v \vert}}{\partial x_v} u(\bsx,\bsalpha) 
	= \sum_{\underset{ k_{d+1} = \vert v \vert - \sum_{i\in v} k_i}{k_v\in \{0,1\}^{\vert v \vert}}}
	c_{v,k_v,k_{d+1}}\;
	u(\bsx,\bsalpha-(k_v;0;k_{d+1}))
	\end{align}
	with
	$
	c_{v,k_v,k_{d+1}} =  (-1)^{k_{d+1}} 
	\prod_{j=1}^{k_{d+1}} (\alpha_{d+1}-j)
	\prod_{i\in v} (\alpha_i-1)^{k_i} .
	$
\end{lemma}
\begin{proof}
	The statement follows by induction over the cardinality of $v$. 
	For $\vert v \vert = 0$, i.e., $v=\emptyset$ both sides of \eqref{al: partial_der}
	are equal to
	$u(\bsx,\bsalpha)$.
	
	Assume $\vert v \vert =1$, i.e., for some $s\in [d]$ we have $v=\{s\}$. 
	Then
	\begin{equation*}
	\frac{\partial}{\partial x_s} u(\bsx,\bsalpha)
	= (\alpha_s-1)u(\bsx,\bsalpha - \bse_s) - (\alpha_{d+1}-1) u(\bsx,\bsalpha-\bse_{d+1}),
	\end{equation*}
	with $\bse_i=(0,\dots,0,1,0,\dots,0)\in \mathbb{R}^{d+1}$ where the $i$th entry is ``1''.
	On the other hand 
	\begin{align*}
	\sum_{\underset{k_{d+1}=1-k_s}{k_s\in \{0,1\}}}
	& c_{\{s\},k_s,k_{d+1}} u(\bsx,\bsalpha-(k_s;0;k_{d+1}))\\
	& =  c_{\{s\},0,1} u(\bsx,  \bsalpha-\bse_{d+1}) 
	+ c_{\{s\},1,0} u(\bsx,\bsalpha-\bse_s).
	\end{align*}
	By the fact that $c_{\{s\},0,1}=-(\alpha_{d+1}-1)$ and $c_{\{s\},1,0}=(\alpha_s-1)$ 
	the claim is proven for $\vert v \vert=1$.
	
	Now assume that \eqref{al: partial_der} is true for any $ v \subseteq [d]$
	with $\vert v \vert \le \ell <d$. Let  $v \subseteq [d]$ with $\vert v \vert = \ell$ be an arbitrary subset  and let $r\in [d]$ with $r\not \in v$. Then we prove that the result also holds
	for $\widetilde v = v \cup \{r\}$.
	We have
	\begin{align*}
	\frac{\partial^{\vert \widetilde v \vert}}{\partial x_{\widetilde v}} u(\bsx,\bsalpha)
	& = \frac{\partial}{\partial x_{r}} \frac{\partial^{\vert v \vert}}{\partial x_{v}} 
	u(\bsx,\bsalpha) \\
	&  = \sum_{\underset{k_{d+1}=\vert v \vert - \sum_{i\in v}k_i}{k_v\in \{0,1\}^{\vert v \vert}}} 
	c_{v,k_{v},k_{d+1}} \frac{\partial }{\partial x_r} u(\bsx,\bsalpha-(k_v;0;k_{d+1})).
	\end{align*}
	Observe that
	\begin{align*}
	&\frac{\partial }{\partial x_r} u(\bsx,\bsalpha-(k_v;0;k_{d+1}))\\
	&= (\alpha_r-1)u(\bsx,\bsalpha-(k_v;0;k_{d+1})-\bse_r)
	-(\alpha_{d+1}-k_{d+1}-1)u(\bsx,\bsalpha-(k_v;0;k_{d+1}+1))\\
	&=\sum_{k_r\in \{0,1\}} (\alpha_r-1)^{k_r}(-1)^{1-k_r}(\alpha_{d+1}-k_{d+1}-1)^{1-k_r}
	u(\bsx,\bsalpha-(k_v;0;\bar k_{d+1})-k_r\bse_{r}),
	\end{align*}
	where $\bar k_{d+1} := k_{d+1}+1-k_r$.
	Further, note that
	\begin{align*}
	& c_{v,k_v,k_{d+1}} (\alpha_r-1)^{k_r} (-1)^{1-k_r} (\alpha_{d+1}-k_{d+1}-1)^{1-k_r}\\
	= & (-1)^{k_{d+1}+1-k_r} 
	\prod_{j=1}^{k_{d+1}} (\alpha_{d+1}-j)(\alpha_{d+1}-k_{d+1}-1)^{1-k_r}
	\prod_{i\in v} (\alpha_i-1)^{k_i}(\alpha_r-1)^{k_r}\\
	= & (-1)^{\bar k_{d+1}} \prod_{j=1}^{\bar k_{d+1}} (\alpha_{d+1}-j) 
	\prod_{i\in \widetilde v}(\alpha_i-1)^{k_i} 
	= c_{\widetilde v, k_{\widetilde v},\bar k_{d+1}}.	  
	\end{align*}
	Hence, by using $\bar k_{d+1} := k_{d+1}+1-k_r$ we obtain
	\begin{align*}
	\frac{\partial^{\vert \widetilde v \vert}}{\partial x_{\widetilde v}} u(\bsx,\bsalpha)
	& = 
	\sum_{\underset{k_{d+1}=\vert v \vert - \sum_{i\in v}k_i}{k_v\in \{0,1\}^{\vert v \vert}}} 
	\sum_{k_r\in\{0,1\}}
	c_{\widetilde v,k_{\widetilde v},\bar k_{d+1}} 
	u(\bsx,\bsalpha-(k_v;0;\bar k_{d+1})-k_r\bse_{r})\\
	& = \sum_{\underset{\bar k_{d+1}=\vert \widetilde v \vert - \sum_{i\in \widetilde v}k_i}
		{k_{\widetilde v}\in \{0,1\}^{\vert \widetilde v \vert}}} 
	c_{\widetilde v,k_{\widetilde v},\bar k_{d+1}} 
	u(\bsx,\bsalpha-(k_{\widetilde v};0;\bar k_{d+1})),
	\end{align*}
	and the proof is finished.
\end{proof}

An immediate consequence of the previous lemma and a chain rule argument we have
for arbitrary $v\subseteq [d]$, $\bsz\in [0,1]^d$ 
and $T_{\bsz}$, defined as in Theorem~\ref{weightedDiscr}, that
\[
\frac{\partial^{\vert v \vert}}{\partial x_v} u(T_{\bsz}\bsx,\bsalpha) 
= \prod_{i\in v} z_i  
\sum_{\underset{ k_{d+1} = \vert v \vert - \sum_{i\in v} k_i}{k_v\in \{0,1\}^{\vert v \vert}}}
c_{v,k_v,k_{d+1}}\;
u(T_{\bsz}\bsx,\bsalpha-(k_v;0;k_{d+1})).
\]
For $\alpha_i\geq 2$ with $1\leq i\leq d$, $\alpha_{d+1}\geq d$ and arbitrary
$\bsx,\bsz\in[0,1]^d$, holds 
$u(T_{\bsz}\bsx,\bsalpha-(k_v;0;k_{d+1}))\leq 1$, where $v\subseteq [d]$, 
$k_v\in \{0,1\}^{\vert v\vert}$ and $k_{d+1}\in [d]$. Then, it follows that
$
\left\vert \frac{\partial^{\vert v \vert}}{\partial x_v} u(T_{\bsz}\bsx,\bsalpha) \right \vert
\leq C^{(1)}_{d,\bsalpha}<\infty,
$
with a constant $C^{(1)}_{d,\bsalpha}$ depending on $d$ and $\bsalpha$.
Hence, for another constant $C^{(2)}_{d,\bsalpha}$ holds
$
\left \Vert  u(T_{\bsz}\, \cdot,\bsalpha) \right \Vert_{H_1} 
\leq C^{(2)}_{d,\bsalpha} <\infty
$
uniformly in $\bsz\in [0,1]^d$. Finally, by the fact
that $u(\bsx) \le 1$
we obtain the following corollary.
\begin{corollary}
	For $\alpha_i\geq 2$ with $1\leq i\leq d$ and $\alpha_{d+1}\geq d$ we have for $u(\bsx,\bsalpha)$
	defined in \eqref{eq: dirichlet_density} that there is a constant $C_{d,\bsalpha}$ such that
	\[
	\frac{ \Vert u(\cdot,\bsalpha)\Vert_D}{\int_{[0,1]^d} u(\bsx,\bsalpha) \rd \bsx} \leq C_{d,\bsalpha} < \infty.
	\]
	
\end{corollary}
This verifies that the application of Theorem~\ref{thm_int_error} 
and Theorem~\ref{weightedDiscr} is justified. For $\bsw^u$ given by \eqref{eq: weight_u} we obtain
\begin{equation*}
\left|
S(f,u(\cdot,\bsalpha))
- \frac{\sum_{i=1}^n f(\bsx_i)\, u(\bsx_i,\bsalpha)}{\sum_{i=1}^n u(\bsx_i,\bsalpha)}
\right| 
\le 4 \|f\|_{H_1} \; C_{d,\bsalpha} \,D_{\lambda_d}(P_n).
\end{equation*}
Consider $f_{\bsgamma}\colon [0,1]^d \to [0,1]$ with $\bsgamma\in (1,\infty)^{d}$ given by
$
f_{\bsgamma}(\bsx) = 2^{-d} \prod_{i=1}^d x_i^{\gamma_i}.
$
Then, by \eqref{eq: norm_const_dirich} we have
\begin{align*}
S(f_\bsgamma,u(\cdot,\bsalpha)) 
& = \frac{1}{2^d}\cdot \frac{\int_{[0,1]^d} u(\bsx,\alpha_1 +\gamma_1,\dots,\alpha_d+\gamma_d,\alpha_{d+1}) \rd \bsx}{
	\int_{[0,1]^d} u(\bsx,\bsalpha) \rd \bsx} \\
& = \frac{1}{2^d}\cdot \frac{\prod_{i=1}^{d} \Gamma(\alpha_i+\gamma_i)}{\prod_{i=1}^{d} \Gamma(\alpha_i)}
\cdot
\frac{\Gamma(\sum_{i=1}^{d+1}\alpha_i)}{\Gamma(\alpha_{d+1}+\sum_{i=1}^{d}(\alpha_i+\gamma_i))}
\end{align*}
and $\Vert f_\bsgamma \Vert_{H_1}=1$. Since we know $S(f_\bsgamma,u(\cdot,\bsalpha))$ we can run the quasi-Monte Carlo importance sampling
algorithm and plot the error for different $d$ and fixed $\bsalpha$
and $\bsgamma$.

\emph{Numerical experiments.}
Let $\bsgamma= (1,\dots,1) \in \mathbb{R}^d$ and $\bsalpha = (2,\dots,2,d)\in \mathbb{R}^d$.
Here the true expectation of $f_\bsgamma$ according to the distribution 
determined by $u(\cdot,\bsalpha)$ can be further simplified to
$
S(f_{\bsgamma},u(\cdot,\bsalpha)) = \frac{(3d-1)!}{(4d-1)!}.
$
Since for large $d$ this value is very small we plot the normalized error. 
For a given point set $P_n$ it is defined by
\begin{equation}  \label{eq: norm_err}
{\rm error}(P_n) = 
\left|
1
- 
\frac{Q_n(f_\bsgamma,u(\cdot,\bsalpha))}{S(f,u(\cdot,\bsalpha))}
\right|,
\end{equation}
and can be computed exactly. 
Let $H_n$ the first $n$ points of the Halton sequence and note that it is known that
$D_{\lambda_d}(H_n) \leq O\left(\frac{(\log n)^d}{n}\right)$. 
By $S_n$ we denote the first $n$ points of the Sobol sequence. 
For details to those standard quasi-Monte Carlo point sets we refer to \cite{DiPi10}.
We obtain the following plots for $d=2,4,6$.

\begin{figure}[htb]
	\centering
	\begin{minipage}[t]{0.45\linewidth}
		\centering
		\includegraphics[width=\linewidth]{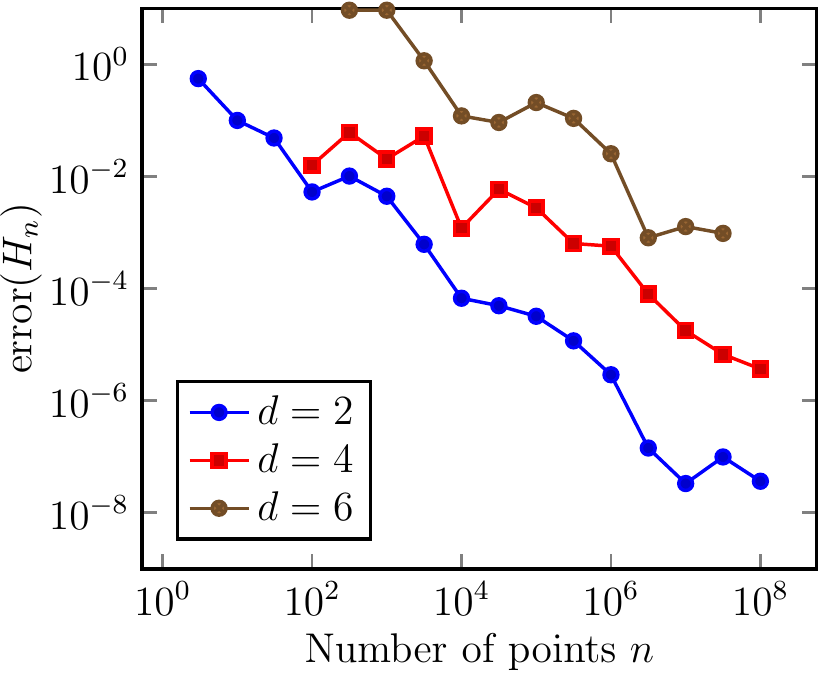}
		\caption{Plot of the normalized error \eqref{eq: norm_err} 
			of $Q_n(f_{\bsgamma},u(\cdot,\bsalpha))$ 
			based on the Halton sequence $H_n$ for $d=2,4,6$.}
	\end{minipage}
	\hfill
	\begin{minipage}[t]{0.45\linewidth}
		\centering
		\includegraphics[width=\linewidth]{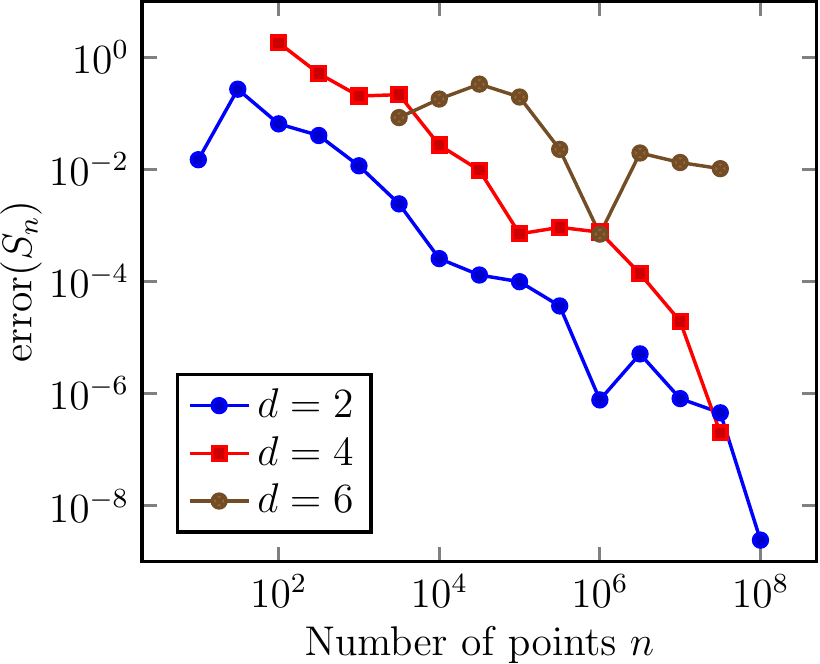}
		\caption{Plot of the normalized error \eqref{eq: norm_err} 
			of $Q_n(f_{\bsgamma},u(\cdot,\bsalpha))$ 
			based on the Sobol sequence $S_n$ for $d=2,4,6$.}
	\end{minipage}
\end{figure}

\subsubsection*{Acknowledgment}
D. Rudolf is supported by
	the Felix-Bernstein-Institute for Mathematical Statistics in the Biosciences, the Campus laboratory AIMS and 
	the DFG within the project 389483880.

\bibliographystyle{amsplain}

\begin{thebibliography}{10}
	
	\bibitem{AgPaSaSt15}
	S.~Agapiou, O.~Papaspiliopoulos, D.~Sanz-Alonso, and A.~Stuart,
	\emph{{Importance {S}ampling: {I}ntrinsic {D}imension and {C}omputational
			{C}ost}}, Statist. Sci. \textbf{32}.
	
	\bibitem{AiDi15}
	Ch. Aistleitner and J.~Dick, \emph{Functions of bounded variation, signed
		measures, and a general {K}oksma--{H}lawka inequality}, Acta Arith.
	\textbf{167} (2015), 143--171.
	
	\bibitem{ACHLT18}
	P.~Arbenz, M.~Cambou, M.~Hofert, C.~Lemieux, and Y.~Taniguchi, \emph{Importance
		sampling and stratification for copula models}, Contemporary Computational
	Mathematics - a celebration of the 80th birthday of Ian Sloan (J. Dick, F. Y.
	Kuo, H. Wo\'zniakowski, eds.), Springer-Verlag, 2018.
	
	\bibitem{ChDi15}
	S.~Chatterjee and P~Diaconis, \emph{{The sample size required in importance
			sampling}}, Ann. Appl. Probab. \textbf{28} (2018), 1099--1135.
	
	\bibitem{DiGaLGiSc16}
	J.~{Dick}, R.~N. {Gantner}, Q.~T. {Le Gia}, and C.~{Schwab}, \emph{{Higher
			order Quasi-Monte Carlo integration for Bayesian Estimation}}, ArXiv e-prints
	(2016).
	
	\bibitem{DiGaLGiSc17}
	\bysame, \emph{Multilevel higher-order quasi-{M}onte {C}arlo {B}ayesian
		estimation}, Math. Models Methods Appl. Sci. \textbf{27} (2017), 953--995.
	
	\bibitem{DiHiPi14}
	J.~Dick, A.~Hinrichs, and F.~Pillichshammer, \emph{Proof {T}echniques in
		{Q}uasi-{M}onte {C}arlo {T}heory}, J. Complexity \textbf{31} (2015),
	327--371.
	
	\bibitem{DiPi10}
	J.~Dick and F.~Pillichshammer, \emph{Digital nets and sequences: {D}iscrepancy
		theory and quasi-{M}onte {C}arlo integration}, Cambridge University Press,
	Cambridge, 2010.
	
	\bibitem{ChGe14}
	M.~{Gerber} and N.~{Chopin}, \emph{{Sequential Quasi-Monte Carlo}}, J. R. Stat.
	Soc. Ser. B. Stat. Methodol. \textbf{77} (2015), 509--579.
	
	\bibitem{HoLe05}
	W.~H{\"o}rmann and J.~Leydold, \emph{Quasi importance sampling}, Preprint,
	Available at {\rm http://epub.wu.ac.at/1394/} (2005).
	
	\bibitem{Jo04}
	G.~Jones, \emph{On the {M}arkov chain central limit theorem}, Probab. Surv.
	\textbf{1} (2004), 299--320.
	
	\bibitem{JoOl10}
	A.~Joulin and Y.~Ollivier, \emph{{C}urvature, concentration and error estimates
		for {M}arkov chain {M}onte {C}arlo}, Ann. Probab. \textbf{38} (2010), no.~6,
	2418--2442.
	
	\bibitem{LaMiNi09}
	K.~{\L}atuszy{\'n}ski, B.~Miasojedow, and W.~Niemiro, \emph{Nonasymptotic
		bounds on the estimation error of {MCMC} algorithms}, Bernoulli \textbf{19}
	(2013), 2033--2066.
	
	\bibitem{MaNo07}
	P.~Math{\'e} and E.~Novak, \emph{{S}imple {M}onte {C}arlo and the {M}etropolis
		algorithm}, J. Complexity \textbf{23} (2007), no.~4-6, 673--696.
	
	\bibitem{NoWo10}
	E.~Novak and H.~Wo{\'z}niakowski, \emph{Tractability of multivariate problems.
		{V}ol. 2: {S}tandard information for functionals}, EMS Tracts in Mathematics,
	vol.~12, European Mathematical Society (EMS), Z{\"u}rich, 2010.
	
	\bibitem{Ow13}
	A.~Owen, \emph{Monte {C}arlo theory, methods and examples}, 2013, in
	preparation.
	
	\bibitem{Pa16}
	D.~Paulin, \emph{Mixing and concentration by {R}icci curvature}, Journal of
	Functional Analysis \textbf{270} (2016), no.~5, 1623 -- 1662.
	
	\bibitem{RoTaFl15}
	V.~Roy, A.~Tan, and J.~Flegal, \emph{{Estimating standard errors for importance
			sampling estimators with multiple {M}arkov chains}}, ArXiv e-prints
	1509.06310 (2015).
	
	\bibitem{Ru09}
	D.~Rudolf, \emph{Explicit error bounds for lazy reversible {M}arkov chain
		{M}onte {C}arlo}, J. Complexity \textbf{25} (2009), no.~1, 11--24.
	
	\bibitem{Ru10}
	\bysame, \emph{Error bounds of computing the expectation by {M}arkov chain
		{M}onte {C}arlo}, Monte Carlo Methods Appl. \textbf{16} (2010), 323--342.
	
	\bibitem{Ru12}
	\bysame, \emph{Explicit error bounds for {M}arkov chain {M}onte {C}arlo},
	Dissertationes Math. \textbf{485} (2012), 93 pp.
	
	\bibitem{RuSp2018}
	D.~Rudolf and B.~Sprungk, \emph{{On a Metropolis-Hastings importance sampling
			estimator}}, ArXiv e-prints (2018).
	
	\bibitem{CoVa10}
	B.~Vandewoestyne and R.~Cools, \emph{On the convergence of quasi-random
		sampling/importance resampling}, Math. Comput. Simulation \textbf{81} (2010),
	no.~3, 490--505.
	
	\bibitem{ViHeFr16}
	M.~Vihola, J.~Helske, and J.~Franks, \emph{Importance sampling type correction
		of {M}arkov chain {M}onte {C}arlo and exact approximations}, ArXiv e-prints
	1609.02541 (2016).
	
\end{thebibliography}

\providecommand{\bysame}{\leavevmode\hbox to3em{\hrulefill}\thinspace}
\providecommand{\MR}{\relax\ifhmode\unskip\space\fi MR }
\providecommand{\MRhref}[2]{%
	\href{http://www.ams.org/mathscinet-getitem?mr=#1}{#2}
}
\providecommand{\href}[2]{#2}

\end{document}